%% file: paper.tex
\documentclass[11pt]{article}
\usepackage{multirow}

\usepackage[utf8]{inputenc} 
\usepackage[T1]{fontenc}    

\usepackage{bbm}
\usepackage{hyperref}       
\usepackage{url}            
\usepackage{booktabs}       
\usepackage{amsfonts}       
\usepackage{nicefrac}       
\usepackage{microtype}      
\usepackage{epsfig,endnotes}
\usepackage{subfig}
\usepackage{graphicx}
\usepackage{caption}
\usepackage{makecell}
\usepackage{moresize}
\usepackage{tikz}
\usepackage{subfloat}
\usepackage{graphicx}
\usepackage{epsfig,endnotes}
\usepackage{multirow}
\usepackage{grffile}
\usepackage[font=bf]{caption}
\usepackage{color, url}
\usepackage{xspace} 

\usepackage{amsmath}

\usepackage{mathrsfs}
\usepackage{amssymb}
\usepackage{amsmath}
\usepackage{amsthm}
\usepackage{epstopdf}
\usepackage{balance}

\newtheorem{theorem}{Theorem}
\usepackage{thm-restate,thmtools}

\newcommand{\myparatight}[1]{\smallskip\noindent{\bf {#1}:}~}
\allowdisplaybreaks

\usepackage[sort&compress,round,comma,authoryear]{natbib}

\usepackage{mathtools}

\usepackage{cleveref}
\usepackage{geometry}
\usepackage{changepage}

\usepackage[ruled,linesnumbered,noend]{algorithm2e}
\usepackage{algorithmic}

\usepackage{pbox}

\usepackage{booktabs}
\newcommand{\argmax}{\operatornamewithlimits{argmax}}
\newcommand{\argmin}{\operatornamewithlimits{argmin}}

\usepackage{eso-pic}

\begin{document}

\AddToShipoutPictureBG*{%
  \AtPageUpperLeft{%
    \setlength\unitlength{1in}%
    \hspace*{\dimexpr0.5\paperwidth\relax}
    \makebox(0,-0.75)[c]{In AAAI Conference on Artificial Intelligence, 2022.}%
}}

%
\title{\LARGE{\bf{Certified Robustness of Nearest Neighbors against \\Data Poisoning and Backdoor Attacks}}}
\author{Jinyuan Jia, Yupei Liu, Xiaoyu Cao, Neil Zhenqiang Gong \\
{Duke University} \\
{\{jinyuan.jia, yupei.liu, xiaoyu.cao, neil.gong\}@duke.edu}
} 

\date{}
\maketitle
\input{abstract}
\input{introduction}
\input{problem}

\input{method}

\input{evaluation}

\input{related}
\input{conclusion}

\section{ Acknowledgments}
We thank the anonymous reviewers for insightful reviews. 
This work was supported by the National Science
Foundation under Grants No. 1937786 and 2112562, as well as the Army Research Office under Grant No. W911NF2110182.

{ 
\bibliographystyle{plainnat}
\bibliography{refs}
}
\input{appendix}


\end{document}

%% file: abstract.tex
\begin{abstract}
Data poisoning attacks and backdoor attacks aim to corrupt a machine learning classifier via modifying, adding, and/or removing some carefully selected training examples, such that the corrupted classifier makes incorrect predictions as the attacker desires. The key idea of state-of-the-art certified defenses against data poisoning attacks and backdoor attacks is to create a \emph{majority vote} mechanism to predict the label of a testing example. Moreover, each voter is a base classifier trained on a subset of the training dataset. Classical simple learning algorithms such as $k$ nearest neighbors (kNN) and radius nearest neighbors (rNN) have intrinsic majority vote mechanisms. In this work, we show that the intrinsic majority vote mechanisms in kNN and rNN already provide certified robustness guarantees against data poisoning attacks and backdoor attacks. Moreover, our evaluation results on MNIST and CIFAR10 show that the intrinsic certified robustness guarantees of kNN and rNN outperform those provided by state-of-the-art certified defenses. Our results serve as standard baselines for future certified defenses against data poisoning attacks and backdoor attacks. 
\end{abstract}

%% file: introduction.tex
\section{Introduction}
Data poisoning attacks and backdoor attacks~\citep{barreno2006can,nelson2008exploiting,biggio2012poisoning,biggio2013security,xiao2015support,steinhardt2017certified,gu2017badnets,chen2017targeted,liu2017trojaning,shafahi2018poison} aim to corrupt the training phase of a machine learning system via carefully poisoning its training dataset including modifying, adding, and/or removing some training examples. Specifically, in data poisoning attacks, the corrupted downstream classifier makes incorrect predictions for clean testing inputs; and in backdoor attacks, the corrupted downstream classifier makes incorrect predictions for testing inputs embedded with a certain trigger.
Data poisoning attacks and backdoor attacks pose severe security concerns to machine learning in critical application domains such as autonomous driving~\citep{gu2017badnets}, cybersecurity~\citep{rubinstein2009antidote,suciu2018does,chen2017targeted}, and healthcare analytics~\citep{mozaffari2014systematic}. 


Multiple certifiably robust learning algorithms~\citep{ma2019data,rosenfeld2020certified,levine2020deep,jia2020intrinsic} against data poisoning attacks  and backdoor attacks were recently developed.  A learning algorithm is certifiably robust against data poisoning attacks and backdoor attacks if it can  learn a classifier on a training dataset that achieves a \emph{certified accuracy} on a testing dataset when the number of poisoned training examples is no more than a threshold (called \emph{poisoning size}).  
The certified accuracy of a learning algorithm is a lower bound of the accuracy of its learnt classifier no matter how an attacker poisons the training examples with the given poisoning size.

The key idea of state-of-the-art certifiably robust learning algorithms~\citep{jia2020intrinsic,levine2020deep} is to create a \emph{majority vote} mechanism to predict the label of a testing example. In particular, each voter votes a label for a testing example and the final predicted label is the majority vote among multiple voters.   
 For instance, Bagging~\citep{jia2020intrinsic} learns multiple base classifiers (i.e., voters), where each of them is learnt on a random subsample of the training dataset.  
 Deep Partition Aggregation (DPA)~\citep{levine2020deep} divides the training dataset into disjoint partitions and learns a base classifier (i.e., a voter) on each partition. 
We denote by $a$ and $b$ the labels with the largest and second largest number of votes, respectively. Moreover, $s_a$ and $s_b$ respectively are the number of votes for labels $a$ and $b$ when there are no corrupted voters. 
 The corrupted voters change their votes from $a$ to $b$ in the worst-case scenario. Therefore, the majority vote result (i.e., the predicted label for a testing example) remains to be $a$ when  the number of corrupted voters is no larger than $\lceil \frac{s_a - s_b}{2} \rceil - 1$.  In other words, the number of corrupted voters that a majority vote mechanism can tolerate depends on the gap $s_a - s_b$ between the largest and the second largest number of votes.

 However, state-of-the-art  certifiably robust learning algorithms achieve suboptimal certified accuracies due to two key limitations. First, 
\emph{each poisoned training example leads to multiple corrupted voters} in the worst-case scenarios. In particular,  modifying a training example corrupts the voters whose training subsamples include the modified training example in bagging~\citep{jia2020intrinsic} and corrupts two voters (i.e., two base classifiers) in DPA~\citep{levine2020deep}.  Therefore, given the same gap $s_a - s_b$ between the largest and the second largest number of votes, the majority vote result is robust against a small number of poisoned training examples. 
Second, \emph{they can only certify robustness for each testing example individually} because it is hard to quantify how poisoned training examples corrupt the voters for different testing examples jointly.  Suppose the classifier learnt by a learning algorithm can correctly classify testing inputs $\mathbf{x}_1$ and $\mathbf{x}_2$. An attacker can poison $e$ training examples such that the learnt classifier misclassifies $\mathbf{x}_1$ or $\mathbf{x}_2$, but the attacker cannot poison $e$ training examples such that both $\mathbf{x}_1$ and $\mathbf{x}_2$ are misclassified. When the poisoning size is $e$, existing certifiably robust learning algorithms would produce a certified accuracy of 0 for the two testing examples. However, the certified accuracy can be 1/2 if we consider them jointly. We note that \citet{steinhardt2017certified} derives an approximate upper bound of the loss function under data poisoning attacks. However, their method cannot certify  the learnt model predicts the same label for a testing example.    

\begin{figure}[!t]
	 \centering
\includegraphics[width=0.8\textwidth]{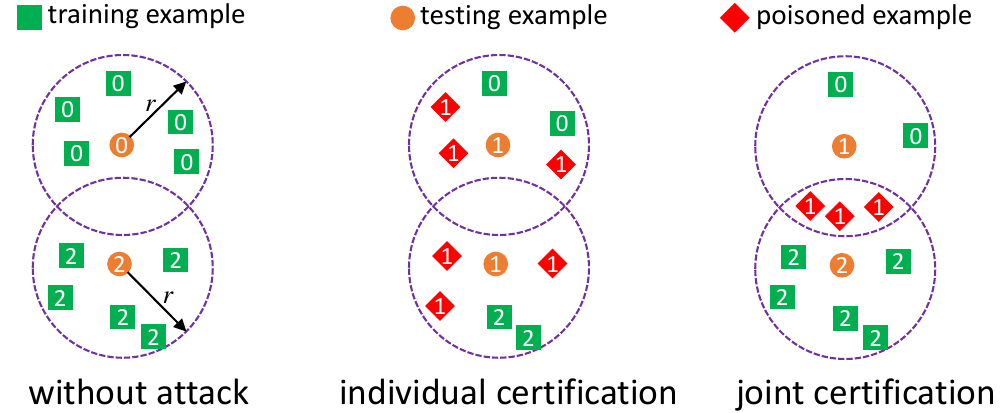}
	 \caption{An example to illustrate individual certification vs. joint certification. Suppose rNN  correctly classifies the two testing examples without attack.  An attacker can poison $3$ training examples. The attacker can make rNN misclassify each testing example individually. However, the attacker cannot make rNN misclassify both testing examples jointly.}
	 \label{illustration_of_joint}
	 \vspace{-3mm}
\end{figure}

 $k$ nearest neighbors (kNN) and radius nearest neighbors (rNN)~\citep{fix1951discriminatory,cover1967nearest} are well-known classic learning algorithms. With good feature representation (e.g., those learnt via self-supervised learning), kNN and rNN can achieve classification accuracy comparable to those of complex learning algorithms such as neural networks~\citep{he2020momentum}. kNN and rNN have intrinsic majority vote mechanisms.
 Specifically, given a testing example, kNN (or rNN) predicts its label via taking a majority vote among the labels of its $k$ nearest neighbors (or neighbors within radius $r$) in the training dataset. Our major contribution in this work is that we show the intrinsic majority vote mechanisms in kNN and rNN make them certifiably robust against data poisoning attacks. Moreover, kNN and rNN address the limitations of state-of-the-art certifiably robust learning algorithms. Specifically, each poisoned training example leads to only one corrupted voter in the  worst-case scenario in kNN and rNN. Thus, given the same gap $s_a - s_b$, the majority vote result (i.e., predicted label for a testing example) is robust against more poisoned training examples in kNN and rNN. 
 
 Furthermore, we show that rNN enables joint certification of multiple testing examples. Figure~\ref{illustration_of_joint} illustrates an example of individual certification and joint certification with two testing examples in rNN. When we treat the two testing examples individually, an attacker can poison 3 training examples such that rNN misclassifies each of them. However,  when we treat them jointly, an attacker cannot poison 3 training examples to misclassify both of them. We propose such joint certification to derive a better certified accuracy for rNN. Specifically, we design methods to group testing examples in a testing dataset such that we can perform joint certification for each group of testing examples. 
  
  We evaluate our methods on MNIST and CIFAR10 datasets. We use $\ell_1$ distance metric to calculate nearest neighbors. First, our methods substantially outperform state-of-the-art certifiably robust learning algorithms. For instance, in data poisoning attacks,  when an attacker can arbitrarily poison 1,000 training examples on MNIST, the certified accuracy of rNN with $r=4$ is 22.9\% and 40.8\% higher than those of bagging~\citep{jia2020intrinsic} and DPA~\citep{levine2020deep}, respectively. Second, our joint certification improves certified accuracy. 
For instance,  our joint certification improves the certified accuracy of rNN by 15.1\%  when an attacker can arbitrarily poison 1,000 training examples on MNIST for data poisoning attacks. Third, we show that self-supervised learning can improve the certified accuracy of kNN and rNN. For instance, when an attacker can arbitrarily poison 500 training examples on CIFAR10 in data poisoning attacks, the certified accuracy of kNN improves by 43.0\% if we use CLIP~\citep{radford2learning}, a feature extractor pre-trained via self-supervised learning, to extract features for each training or testing input.

In summary, we make the following contributions: 
\begin{itemize}
    \item We derive the intrinsic certified robustness guarantees of kNN and rNN against data poisoning attacks and backdoor attacks. 
    \item We propose joint certification of multiple testing examples to derive a better certified robustness guarantee for rNN. rNN is the first method that supports joint certification of multiple testing examples. 
    \item We evaluate our methods and compare them with state-of-the-art on  MNIST and CIFAR10. 
\end{itemize}

%% file: problem.tex
\section{Problem Setup}
\vspace{-2mm}
\myparatight{Learning setting} 
Assuming we have a training dataset $D_{tr}$ with $n$ training examples. 
We denote by $\mathcal{M}$ a learning algorithm. Moreover, we denote by $\mathcal{M}(D_{tr}, \mathbf{x})$ the label predicted for a testing input $\mathbf{x}$ by a classifier learnt by $\mathcal{M}$ on the training dataset $D_{tr}$. 
For instance,  
given a training dataset $D_{tr}$ and a testing input $\mathbf{x}$, kNN finds the $k$ training examples in $D_{tr}$ that are the closest to $\mathbf{x}$ as the nearest neighbors, while rNN finds the training examples  in $D_{tr}$ whose distances to  $\mathbf{x}$ are no larger than $r$ as the nearest neighbors. The distance between a training input and a testing input can be measured by any distance metric. 
Then, kNN and rNN use majority vote among the nearest neighbors to predict the label of $\mathbf{x}$. Specifically, each nearest neighbor is a voter and votes its label for the testing input $\mathbf{x}$; and the label with the largest number of votes is the final predicted label for $\mathbf{x}$.

\myparatight{Data poisoning attacks} 
We consider  data poisoning attacks~\citep{rubinstein2009antidote,biggio2012poisoning,xiao2015feature,li2016data,munoz2017towards,jagielski2018manipulating} that aim to 
 poison (i.e., modify, add, and/or remove) some carefully selected training examples in $D_{tr}$ such that the corrupted   classifier has a low accuracy for testing inputs (either indiscriminate clean testing inputs or attacker-chosen ones).

 \myparatight{Backdoor attacks} In backdoor attacks~\citep{gu2017badnets,liu2017trojaning,chen2017targeted}, an attacker also poisons the training dataset, but the corrupted classifier makes incorrect, attacker-chosen predictions for testing inputs embedded with a certain trigger. For instance, the attacker can embed the trigger to some training inputs in $D_{tr}$ and relabel them as the attacker-chosen label. The classifier built based on such poisoned training dataset predicts the attacker-chosen label for any testing input embedded with the same trigger. However, the predictions for clean testing inputs without the trigger are unaffected, i.e., the corrupted classifier and the clean classifier are highly likely to predict the same label for a clean testing input.

\myparatight{Poisoned training dataset} Both data poisoning attacks and backdoor attacks poison the training dataset to achieve their goals. For simplicity, we use $D^{*}_{tr}$ to denote the \emph{poisoned training dataset}. Note that $D^{*}_{tr}$ could include duplicate training examples, e.g., when the attacker adds duplicate training examples.  Moreover, we define the \emph{poisoning size} of  a poisoned training dataset $D^{*}_{tr}$ (denoted as $S(D_{tr}, D^{*}_{tr})$) as the minimal number of modified/added/removed training examples that can turn $D_{tr}$ into $D_{tr}^{*}$. Formally,  $S(D_{tr}, D^{*}_{tr})=\max\{|D_{tr}^{*}|,|D_{tr}|\}-|D_{tr}^{*}\cap D_{tr}|$ is the poisoning size of $D_{tr}^{*}$.

\myparatight{Certified accuracy}
Given a training dataset $D_{tr}$ and a learning algorithm $\mathcal{M}$,  we use \emph{certified accuracy} on a testing dataset $D_{te}=\{(\mathbf{x}_i, y_i)\}_{i=1}^{t}$ to measure the algorithm's performance. Specifically, 
we denote \emph{certified accuracy at poisoning size $e$} as $CA(e)$ and formally define it as follows:
{\small 
\begin{align}
 CA(e)=  \min_{D^{*}_{tr}, S(D_{tr}, D^{*}_{tr}) \leq e} \frac{\sum_{(\mathbf{x}_i,y_i)\in D_{te}}\mathbb{I}(\mathcal{M}(D^{*}_{tr},\mathbf{x}_i)= y_{i})}{|D_{te}|},  
\end{align}
}
where  $\mathbb{I}$ is the indicator function and $\mathcal{M}(D_{tr}^{*}, \mathbf{x}_i)$ is the label predicted for a testing input $\mathbf{x}_i$ by the classifier learnt by the algorithm $\mathcal{M}$ on the poisoned training dataset $D_{tr}^{*}$. $CA(e)$ is the least testing accuracy on $D_{te}$ that the learning algorithm  $\mathcal{M}$ can achieve no matter how an attacker poisons the training examples when the poisoning size is at most $e$.  For data poisoning attacks, the testing dataset $D_{te}$ is a set of clean testing examples. For backdoor attacks, $D_{te}$ is a set of testing examples embedded with a trigger. 
Our goal is to derive  lower bounds of $CA(e)$ for learning algorithms kNN and rNN.

%% file: method.tex
\section{Certified Accuracy of kNN and rNN}
We first derive a lower bound of the certified accuracy via \emph{individual certification}, which treats testing examples in $D_{te}$ individually. Then, we derive a better lower bound of the certified accuracy for rNN via \emph{joint certification}, which treats testing examples jointly.  

\subsection{Individual Certification}
Given a poisoning size at most $e$, our idea is to certify whether the predicted label stays unchanged or not for each testing input individually. If the predicted label of a testing input $\mathbf{x}$ stays unchanged (i.e., $\mathcal{M}(D_{tr}, \mathbf{x})=\mathcal{M}(D_{tr}^*, \mathbf{x})$) and it matches with the testing input's true label, then kNN or rNN certifiably correctly classifies the testing input when the poisoning size is at most $e$. Therefore, we can obtain a lower bound of the certified accuracy at poisoning size $e$ as the fraction of testing inputs in $D_{te}$  which kNN or rNN certifiably correctly classifies. Next, we first discuss how to certify whether the predicted label stays unchanged or not for each testing input individually. Then, we show our lower bound of the certified accuracy at poisoning size $e$.

\myparatight{Certifying the predicted label of a testing input} Our goal is to certify that $\mathcal{M}(D_{tr}, \mathbf{x})=\mathcal{M}(D_{tr}^*, \mathbf{x})$ for a testing input $\mathbf{x}$ when the poisoning size is no larger than a threshold. 
 Given a training dataset $D_{tr}$ (or a poisoned training dataset $D_{tr}^*$) and a testing input $\mathbf{x}$, we use $\mathcal{N}(D_{tr},\mathbf{x})$ (or $\mathcal{N}(D_{tr}^*,\mathbf{x})$) to denote the set of nearest neighbors of $\mathbf{x}$ in $D_{tr}$ (or $D_{tr}^*$) for kNN or rNN.  
We note that there may exist ties when determining the nearest neighbors for kNN, 
i.e., multiple training examples may have the same distance to the testing input. Usually, kNN breaks such ties uniformly at random. However, such random ties breaking method introduces randomness, i.e., the difference of nearest neighbors before and after poisoned training examples (i.e., $\mathcal{N}(D_{tr},\mathbf{x})$ vs. $\mathcal{N}(D_{tr}^*,\mathbf{x})$) depends on the randomness in breaking ties. 
Such randomness makes it challenging to certify the robustness of the predicted label against poisoned training examples. To address the challenge, we propose to define a deterministic ranking of training examples and break ties via choosing the training examples with larger ranks. Moreover, such ranking between clean training examples does not depend on poisoned ones. For instance, we can use a cryptographic hash function (e.g., SHA-1) that is very unlikely to have collisions to hash each training example based on its input feature vector and label, and then we rank the training examples based on their hash values.

We use $s_l$ to denote the number of votes in $\mathcal{N}(D_{tr},\mathbf{x})$ for label $l$, i.e., the number of nearest neighbors in $\mathcal{N}(D_{tr},\mathbf{x})$ whose labels are $l$. Formally, we have $s_l = \sum_{(\mathbf{x}_j,y_j) \in \mathcal{N}(D_{tr},\mathbf{x})} \mathbb{I}(y_j = l)$,  
where $l=1,2,\cdots,c$ and $\mathbb{I}$ is an indicator function. We note that $s_l$ also depends on the testing input $\mathbf{x}$. However, we omit the explicit dependency on $\mathbf{x}$ for simplicity.  kNN or rNN essentially predicts the label of the testing input $\mathbf{x}$ as the label with the largest number of votes, i.e., $\mathcal{M}(D_{tr}, \mathbf{x})=\argmax_{l\in \{1,2,\cdots,c\}}s_l$. 
Suppose $a$ and $b$ are the labels with the largest and second largest number of votes, i.e., $s_a$ and $s_b$ are the largest and second largest ones among $\{s_1,s_2,\cdots,s_c\}$, respectively. We note that there may exist ties when comparing the labels based on their votes. We define a deterministic ranking of labels in $\{1,2,\cdots,c\}$ and take the label with the largest rank when such ties happen. For instance, when labels 2 and 3 have tied largest number of votes, we take label 3 as $a$. In the worse-case scenario, each poisoned training example leads to one corrupted voter in kNN or rNN, which changes its vote from label $a$ to label $b$. Therefore, kNN or rNN still predicts label $a$ for the testing input $\mathbf{x}$ when the number of poisoned training examples is no more than  $\lceil \frac{s_a - s_b}{2} \rceil - 1$ (without considering the ties breaking). Formally, we have the following theorem:
\begin{theorem}
\label{nn_certified_theorem}
Assuming we have a training dataset $D_{tr}$, a testing input $\mathbf{x}$, and a nearest neighbor algorithm $\mathcal{M}$ (i.e., kNN or rNN). $a$ and $b$ respectively are the two labels with the largest and second largest number of votes among the nearest neighbors $\mathcal{N}(D_{tr},\mathbf{x})$ of $\mathbf{x}$ in $D_{tr}$. Moreover, $s_a$ and $s_b$ are the number of votes for $a$ and $b$, respectively. Then, we have the following: 
\begin{align}
&\mathcal{M}(D_{tr}^*, \mathbf{x}) = a, \nonumber \\
&\forall D^{*}_{tr}  \text{ such that } S(D_{tr}, D_{tr}^*) \leq \lceil \frac{s_a - s_b + \mathbb{I}(a>b)}{2} \rceil - 1. 
\end{align}
\end{theorem}
\begin{proof}
See Supplementary Material.
\end{proof}

\myparatight{Deriving a lower bound of $CA(e)$} kNN or rNN certifiably correctly classifies a testing input $\mathbf{x}$ if it correctly predicts its label before attacks and the predicted label stays unchanged after an attacker poisons the training dataset. Therefore, the fraction of testing inputs that kNN or rNN certifiably correctly classifies is a lower bound of $CA(e)$. Formally, we have the following theorem:  

\begin{theorem}[Individual Certification]
\label{theorem_of_ic_aggregate}
Assuming we have a training dataset $D_{tr}$, a testing dataset  $D_{te}=\{(\mathbf{x}_i,y_i)\}_{i=1}^{t}$, and a nearest neighbor algorithm $\mathcal{M}$ (i.e., kNN or rNN).  $a_i$ and $b_i$ respectively are the two labels with the largest and second largest number of votes among the nearest neighbors $\mathcal{N}(D_{tr},\mathbf{x}_i)$ of $\mathbf{x}_i$ in $D_{tr}$. Moreover, $s_{a_i}$  and $s_{b_i}$ are the number of votes for $a_i$ and $b_i$, respectively. Then, we have the following lower bound of $CA(e)$: 
\begin{align}
 {CA}(e) \geq \frac{\sum_{(\mathbf{x}_i,y_i)\in D_{te}} \mathbb{I}(a_i = y_i)\cdot \mathbb{I}(e \leq e_i^{*})}{|D_{te}|}, 
\end{align}
where  $e_{i}^{*}= \lceil \frac{s_{a_i} - s_{b_i} + \mathbb{I}(a_i > b_i)}{2} \rceil - 1$. 
\end{theorem}
\begin{proof}
See Supplementary Material.

\end{proof}

\subsection{Joint Certification}
We derive a better lower bound of the certified accuracy via jointly considering multiple testing examples. 
Our intuition is that, given a group of testing examples and a poisoning size $e$, an attacker may not be able to make a learning algorithm misclassify all the testing examples jointly even if it can make the learning algorithm misclassify each of them individually. In particular, rNN enables such joint certification. It is challenging to perform joint certification for kNN because of the complex interactions between the nearest neighbors of different testing examples (see our proof of Theorem~\ref{nn_certified_theorem_opt} for specific reasons). Next, we first derive a lower bound of $CA(e)$ on  a group of testing examples for rNN. Then, we derive a lower bound of $CA(e)$ on the testing dataset $D_{te}$ via dividing it into groups. Finally, we discuss different strategies to divide the testing dataset into groups, which may lead to different lower bounds of $CA(e)$.

\myparatight{Deriving a lower bound of $CA(e)$ for a group of testing examples} Suppose we have a group of testing examples which have different predicted labels in rNN. Our key intuition is that 
when an attacker can poison $e$ training examples, the attacker can only decrease the total votes for the testing examples' predicted labels by at most $e$ in rNN, as the testing examples' predicted labels are different. We denote by $\mathcal{U}$ a group of testing examples with different predicted labels and by $m$ its size, i.e., $m=|\mathcal{U}|$. The next theorem shows a lower bound of $CA(e)$ on the testing examples in $\mathcal{U}$ for rNN. 
\begin{theorem}
\label{nn_certified_theorem_opt}
Assuming we have a training dataset $D_{tr}$,  the learning algorithm rNN, and a group of $m$ testing examples $\mathcal{U}=\{(\mathbf{x}_i,y_i)\}_{i=1}^{m}$ with different predicted labels. 
$a_i$ and $b_i$ respectively are the two labels with the largest and second largest number of votes among the nearest neighbors $\mathcal{N}(D_{tr},\mathbf{x}_i)$ of $\mathbf{x}_i$ in $D_{tr}$. Moreover, $s_{a_i}$  and $s_{b_i}$ are the number of votes for $a_i$ and $b_i$, respectively. 
 Without loss of generality, we assume the following:
\begin{align}
\label{lossless_condition}
&(s_{a_1}-s_{b_1})\cdot \mathbb{I}(a_1 = y_1) \geq (s_{a_2}-s_{b_2})\cdot \mathbb{I}(a_2 = y_2) \geq 
\cdots \geq (s_{a_m}-s_{b_m})\cdot \mathbb{I}(a_m = y_m). 
\end{align}
Then, the certified accuracy at poisoning size $e$ of rNN for $\mathcal{U}$ has a lower bound ${CA}(e) \geq  \frac{w-1}{|\mathcal{U}|}$, where  $w$ is the solution to the following optimization problem: 
\begin{align}
 & w = \argmin_{w', w'\geq 1} w' \nonumber \\
  \label{optimizationw}
 \text{ s.t. } \sum_{i=w'}^{m}& \max(s_{a_i} - s_{b_i} - e+ \mathbb{I}(a_i > b_i),0)  \cdot \mathbb{I}(a_{i}=y_{i}) \leq e .
\end{align}
\end{theorem}
\begin{proof}
When an attacker can poison at most $e$ training examples, the attacker can add at most $e$ new nearest neighbors and remove $e$ existing ones in $\mathcal{N}(D_{tr},\mathbf{x}_i)$ (equivalent to modifying $e$ training examples) in the worst-case scenario. 
We denote by $s^{*}_{a_i}$ and $s^{*}_{b_i}$ respectively the number of votes for labels $a_i$ and $b_i$ among the nearest neighbors $\mathcal{N}(D_{tr}^*,\mathbf{x}_i)$. 
First, we have $s^{*}_{b_i} \leq s_{b_i} + e$ for $\forall i \in \{1,2,\cdots,m\}$ since at most $e$ new nearest neighbors are added. Second, we have $s^{*}_{a_i} \geq s_{a_i} - e_i$ in rNN, where $e_i$ is the number of removed nearest neighbors in $\mathcal{N}(D_{tr},\mathbf{x}_i)$ whose true labels are $a_i$. Note that kNN does not support joint certification because $s^{*}_{a_i} \geq s_{a_i} - e_i$ does not hold for kNN.  

Next, we derive the minimal value of $e_i$ such that rNN misclassifies $\mathbf{x}_i$.  
In particular, we consider two cases. If $a_i \neq y_i$, i.e.,  $\mathbf{x}_i$ is misclassified by rNN without attack, then we have $e_i = 0$. If $a_i = y_i$,  $\mathbf{x}_i$ is misclassified by rNN when $s^{*}_{a_i} \leq s^{*}_{b_i}$ if $a_i < b_i$ and $s^{*}_{a_i} < s^{*}_{b_i}$ if $a_i > b_i$ after attack, which means $e_i \geq s_{a_i} - s_{b_i} - e + \mathbb{I}(a_i > b_i)$. Since $e_i \geq 0$, we have $e_i \geq \max (s_{a_i} - s_{b_i} - e+ \mathbb{I}(a_i > b_i),0)$. Combining the two cases, we have the following lower bound for $e_i$  that makes rNN misclassify $\mathbf{x}_i$: $e_i \geq \max (s_{a_i} - s_{b_i} - e+ \mathbb{I}(a_i > b_i), 0) \cdot \mathbb{I}(a_i = y_i)$. 
Moreover, since the attacker can remove at most $e$ training examples and the group of testing examples have different predicted labels, i.e., $a_i\neq a_j$ $\forall i,j \in \{1,2,\cdots,m\}$ and $i\neq j$, we have $\sum_{i=1}^{m} e_i \leq e$. We note that the lower bound of $e_i$ is non-increasing as $i$ increases based on Equation~(\ref{lossless_condition}). Therefore, in the worst-case scenario, the attacker can make rNN misclassify the last $m-w+1$ testing inputs whose corresponding $e_i$ sum to be at most $e$. Formally, $w$ is the solution to the optimization problem in Equation~(\ref{optimizationw}). Therefore, the certified accuracy at poisoning size $e$ is at least $\frac{w-1}{|\mathcal{U}|}$.
\end{proof}

\myparatight{Deriving a lower bound of $CA(e)$ for a testing dataset} Based on Theorem~\ref{nn_certified_theorem_opt}, we can derive a lower bound of $CA(e)$ for a testing dataset via dividing it into disjoint groups, each of which includes testing examples with different predicted labels in rNN. Formally, we have the following theorem:
\begin{theorem}[Joint Certification]
\label{joint_dataset_aggregate_opt}
Given a testing dataset $D_{te}$, we divide it into $\lambda$ disjoint groups, i.e., $\mathcal{U}_1, \mathcal{U}_2, \cdots, \mathcal{U}_{\lambda}$, where the testing examples in each group have different predicted labels in rNN. Then, we have the following lower bound of $CA(e)$: 
\begin{align}
\label{aggregate_certified_accuracy}
    CA(e) \geq \frac{\sum_{j=1}^{\lambda}\mu_j \cdot |\mathcal{U}_j|}{\sum_{j=1}^{\lambda}|\mathcal{U}_j|}, 
\end{align}
where $\mu_j$ is the lower bound of the certified accuracy at poisoning size $e$ on group $\mathcal{U}_j$, which we can obtain by  invoking Theorem~\ref{nn_certified_theorem_opt}. 
\end{theorem}
\begin{proof}
See Supplementary Material.
\end{proof}

\noindent
{\bf Strategies of grouping testing examples:} Our Theorem~\ref{joint_dataset_aggregate_opt} is applicable to any way of dividing the testing examples in $D_{te}$ to disjoint groups once the testing examples in each group have different predicted labels in rNN. Therefore, a natural question is how to group the testing examples in $D_{te}$ to maximize our lower bound of certified accuracy.
For instance, a naive method is to randomly divide the testing examples into disjoint groups, each of which includes at most $c$ (the number of classes) testing examples with different predicted labels. We call such method \emph{Random Division (RD)}. However, RD achieves suboptimal performance because it does not consider the certified robustness of each individual testing example. In particular, some testing examples can or cannot be certifiably correctly classified no matter which groups they belong to. However, if we group them with other testing examples,  the certified accuracy may be degraded because each group can have at most $c$ testing examples. For instance, if a testing example cannot be certifiably correctly classified no matter which group it belongs to, then adding it to a group would exclude another testing example from the group, which may degrade the certified accuracy for the group. 

Therefore, we propose to isolate these testing examples and divide the remaining testing examples into disjoint groups. We call such method \emph{\underline{Is}o\underline{la}tion \underline{a}nd \underline{D}ivision (ISLAND)}.  
Specifically, we first divide the testing dataset $D_{te}$ into three disjoint parts which we denote as $D_{te}^{0}$, $D_{te}^{1}$, and $D_{te}^{2}$. 
$D_{te}^{0}$ contains the testing examples that cannot be certifiably correctly classified at poisoning size $e$ no matter which group they belong to. Based on our proof of Theorem~\ref{nn_certified_theorem_opt}, a testing example $(\mathbf{x}_i, y_i)$ that satisfies $(s_{a_i} - s_{b_i} - e + \mathbb{I}(a_i > b_i))\cdot \mathbb{I}(a_{i}=y_{i}) \leq 0$ cannot be certifiably correctly classified at poisoning size $e$ no matter which group it belongs to. Therefore,  $D_{te}^{0}$ includes such testing examples. Moreover, based on Theorem~\ref{nn_certified_theorem}, a testing example $(\mathbf{x}_i, y_i)$ that satisfies $e \leq \lceil \frac{s_{a_i} - s_{b_i} + \mathbb{I}(a_i > b_i)}{2} \rceil - 1$ can be certifiably correctly classified at poisoning size $e$. Therefore, $D_{te}^{1}$ includes such testing examples. Each testing example in $D_{te}^{0}$ or $D_{te}^{1}$ forms a group by itself.  $D_{te}^{2}$ includes the remaining testing examples, which we further divide into groups. Our method of dividing  $D_{te}^{2}$ into groups is inspired by the proof of Theorem~\ref{nn_certified_theorem_opt}. 
In particular,  
we form a group of testing examples as follows: for each label $l \in \{1,2,\cdots,c\}$, we find the testing example that has the largest value of $(s_{a_i} - s_{b_i}-e+ \mathbb{I}(a_i > b_i))\cdot \mathbb{I}(a_i = l)$ and we skip the label if there is no remaining testing example whose predicted label is $l$.  
We apply the procedure to recursively group the testing examples in $D_{te}^{2}$ until no testing examples are left.

%% file: evaluation.tex
\section{Evaluation}
\noindent
{\bf Datasets:} We evaluate our methods on MNIST and CIFAR10. 
We use the popular histogram of oriented gradients (HOG)~\citep{dalal2005histograms} method (we adopt public implementation~\citep{hog-code}) to extract features for each example, which we found improves certified accuracy. Note that previous work~\citep{jia2020intrinsic} used a pre-trained model to extract features via transfer learning. However, the pre-trained model may also be poisoned and thus we don't use it. For simplicity, we rank the training examples in a dataset using their indices and use them to break ties in determining nearest neighbors for kNN. Moreover, we rank the labels as $\{1,2,\cdots,10\}$ to break ties for labels. Following previous work~\citep{wang2019neural}, we adopt a white square located at the bottom right corner of an image as the trigger in backdoor attacks for both MNIST and CIFAR10. The sizes of the triggers are $5\times 5$ and $10\times 10$ for MNIST and CIFAR10, considering different image sizes in those two datasets.

\begin{figure}[!t]
	 \centering
\subfloat[MNIST]{\includegraphics[width=0.45\textwidth]{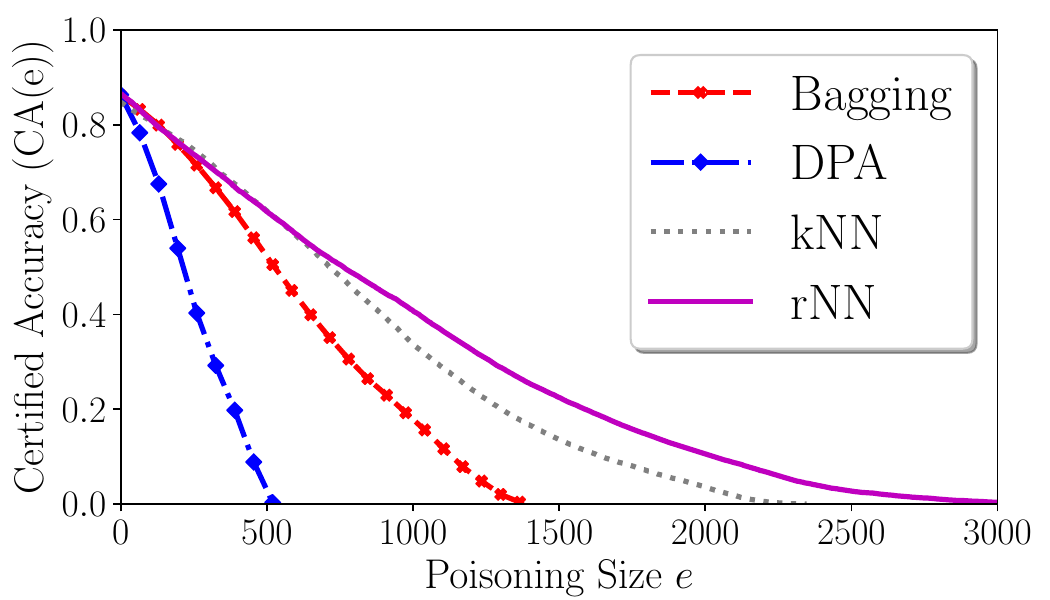}\label{compare_mnist}}
\subfloat[CIFAR10]{\includegraphics[width=0.45\textwidth]{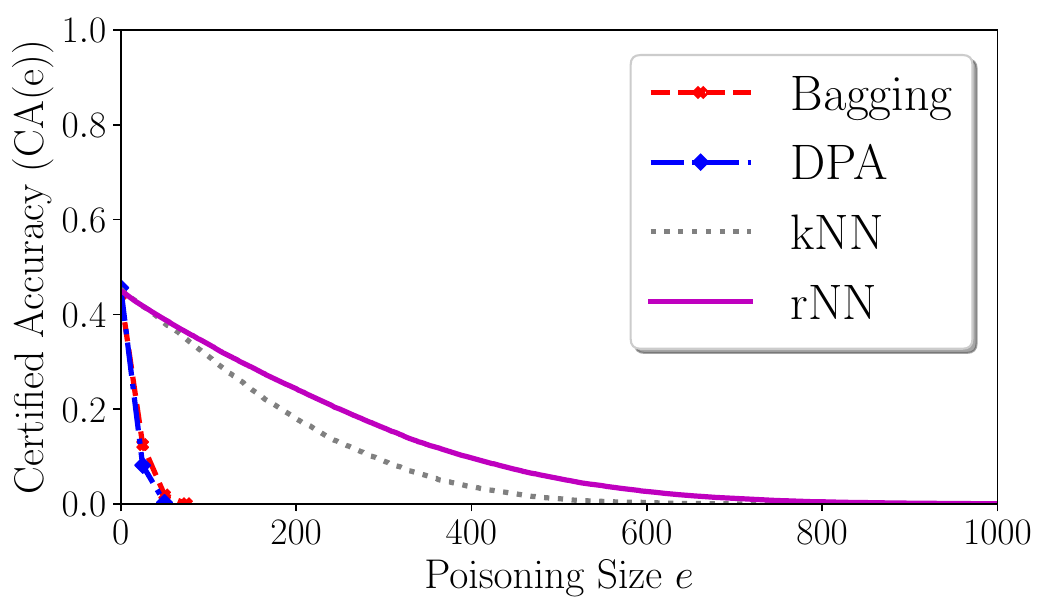}\label{compare_cifar10}}
	 \caption{Comparing kNN and rNN with state-of-the-art methods against data poisoning attacks. }
	 \label{compare_hog_feature_extractor_dp}
	 \vspace{-4mm}
\end{figure}

\begin{figure}[!t]
	 \centering
\subfloat[MNIST]{\includegraphics[width=0.45\textwidth]{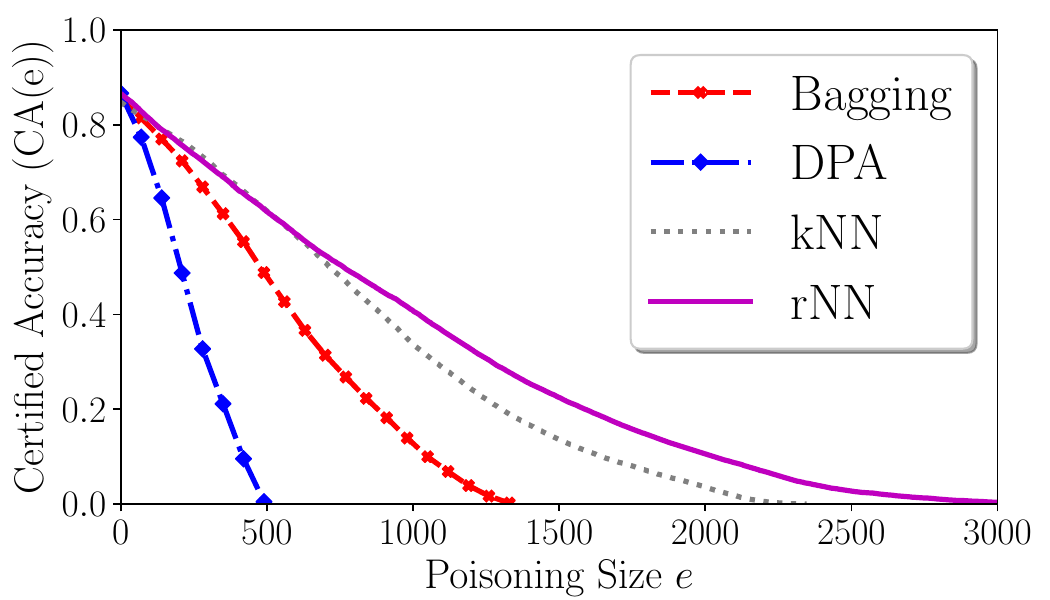}\label{compare_mnist}}
\subfloat[CIFAR10]{\includegraphics[width=0.45\textwidth]{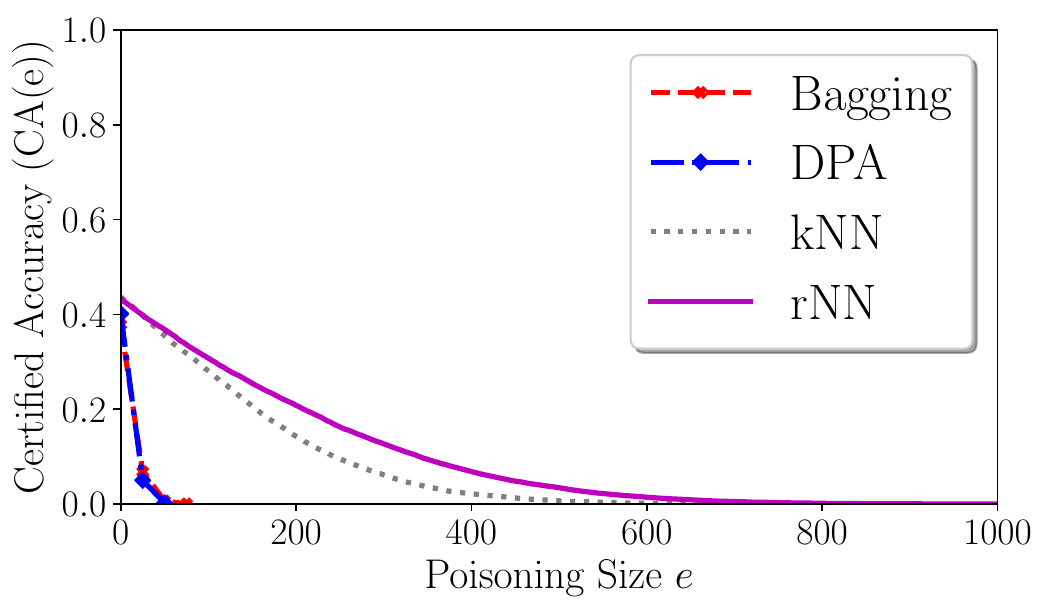}\label{compare_cifar10}}
	 \caption{Comparing kNN and rNN with state-of-the-art methods against backdoor attacks. }
	 \label{compare_hog_feature_extractor_ba}
	 \vspace{-4mm}
\end{figure}

\noindent
{\bf Parameter settings:} While any distance metric is applicable, we use $\ell_1$ in our experiments for both kNN and rNN. Unless otherwise mentioned, we adopt the following settings: $k=5,000$ for both MNIST and CIFAR10 in kNN; and $r=4$ for MNIST and $r=20$ for CIFAR10 in rNN, considering the different feature dimensions of MNIST and CIFAR10. By default, we use the ISLAND grouping method in the joint certification for rNN.

\noindent
{\bf Comparing with  bagging~\citep{jia2020intrinsic} and DPA~\citep{levine2020deep}:} Figure~\ref{compare_hog_feature_extractor_dp} and~\ref{compare_hog_feature_extractor_ba} show the comparison results of bagging, DPA, kNN, and rNN for data poisoning attacks and backdoor attacks, respectively. Bagging learns $N$ base classifiers, each of which is learnt on a random subsample with $\xi$ training examples of the training dataset. Moreover, bagging's certified accuracy is correct with a confidence level $1-\alpha$. DPA divides a training dataset into $\zeta$ disjoint partitions and learns a base classifier on each of them. Then, DPA takes a majority vote among the base classifiers to predict the label of a testing example.  All the compared methods have tradeoffs between accuracy under no attacks (i.e., $CA(0)$) and robustness against attacks. Therefore, we set their parameters such that they have similar accuracy under no attacks (i.e., similar $CA(0)$). In particular, we use the default $k$ for kNN, and we adjust $r$ for rNN, $\xi$ for bagging, and $\zeta$ for DPA. The searched parameters are as follows: $r=4$, $\xi =20$, and $\zeta = 2,500$ for MNIST; and $r=21$, $\xi = 300$, and $\zeta = 500$ for CIFAR10. We find that the searched parameters are the same for data poisoning attacks and backdoor attacks. Note that we set $N=1,000$ and $\alpha =0.001$ for bagging following \citet{jia2020intrinsic}.

\begin{figure}[!t]
	 \centering
\subfloat[MNIST]{\includegraphics[width=0.45\textwidth]{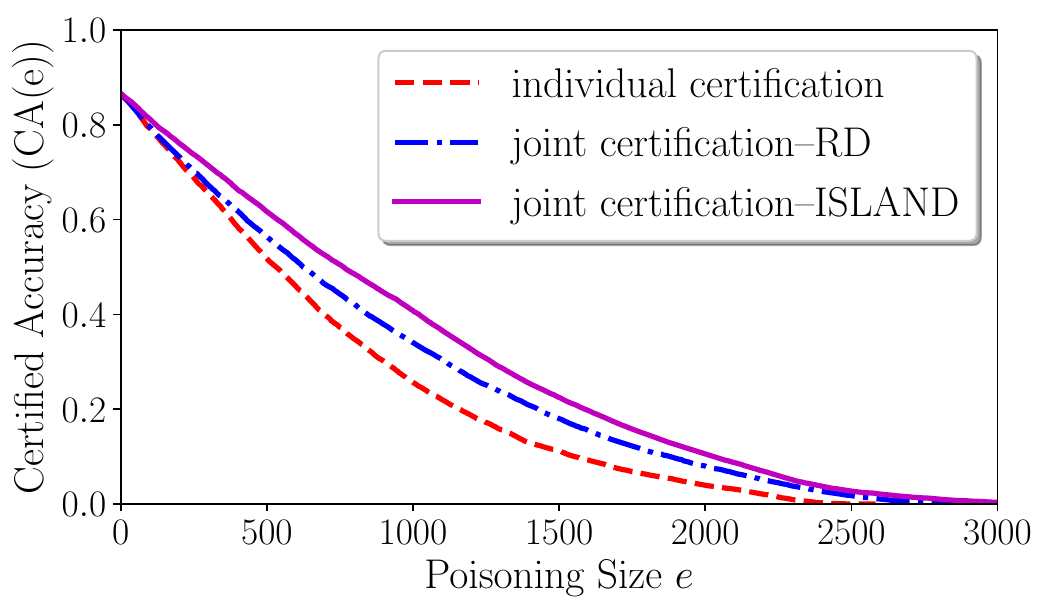}\label{individual_joint_mnist}}
\subfloat[CIFAR10]{\includegraphics[width=0.45\textwidth]{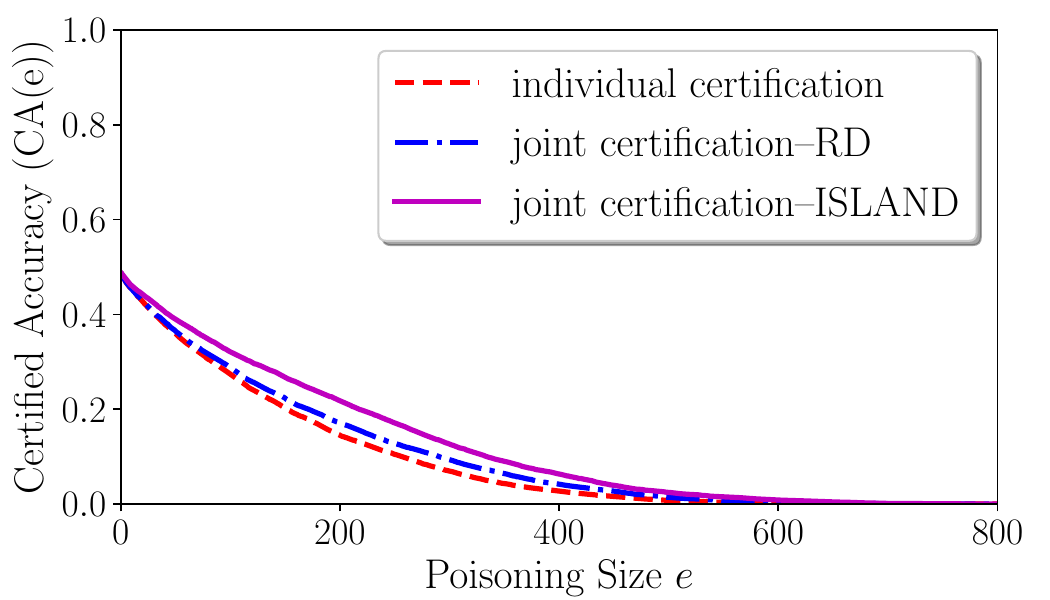}\label{individual_joint_cifar10}}
	 \caption{Comparing individual certification with joint certification for rNN against data poisoning attacks.}
	 \label{compare_hog_feature_extractor_joint}
	 \vspace{-4mm}
\end{figure}

We have the following observations. First, both kNN and rNN outperform bagging and DPA. The superior performance of kNN and rNN stems from two reasons: 1) each poisoned training example corrupts multiple voters for bagging and DPA, while it only corrupts one voter for kNN and rNN,  which means that, given the same gap between the largest and second largest number of votes, kNN and rNN can tolerate more poisoned training examples; and 2) rNN enables joint certification that improves the certified accuracy. Second,  rNN achieves better certified accuracy than kNN when the poisoning size is large.  The reason is that rNN supports joint certification. Third, kNN (or rNN) achieves similar certified accuracy against data poisoning attacks and backdoor attacks for a given poisoning size. This is because, in backdoor attacks, adding a trigger to a testing image does not affect its label predicted by a clean classifier. 

\noindent
{\bf Comparing individual certification with joint certification:} Figure~\ref{compare_hog_feature_extractor_joint} and~\ref{compare_hog_feature_extractor_joint_ba} (in Supplementary Material) compare individual certification and joint certification (with the RD and ISLAND grouping methods) for rNN against data poisoning attacks and backdoor attacks. Our empirical results validate that joint certification improves the certified accuracy upon individual certification. Moreover, our ISLAND grouping method outperforms the RD method. 

\noindent
{\bf Impact of $k$ and $r$:} Figure~\ref{impact_of_k_knn},~\ref{impact_of_r_rnn},~\ref{impact_of_k_knn_ba} (in Supplementary Material),~\ref{impact_of_r_rnn_ba} (in Supplementary Material) show the impact of $k$ and $r$ on the certified accuracy of kNN and rNN against data poisoning attacks and backdoor attacks, respectively. 
As the results show, 
$k$ and $r$ achieve tradeoffs between accuracy under no attacks (i.e., $CA(0)$) and robustness. Specifically, 
when $k$ or $r$ is smaller, the accuracy under no attacks, i.e., $CA(0)$, is larger, but the certified accuracy decreases more quickly as the poisoning size $e$ increases.

\begin{figure}[!t]
	 \centering
\subfloat[MNIST]{\includegraphics[width=0.45\textwidth]{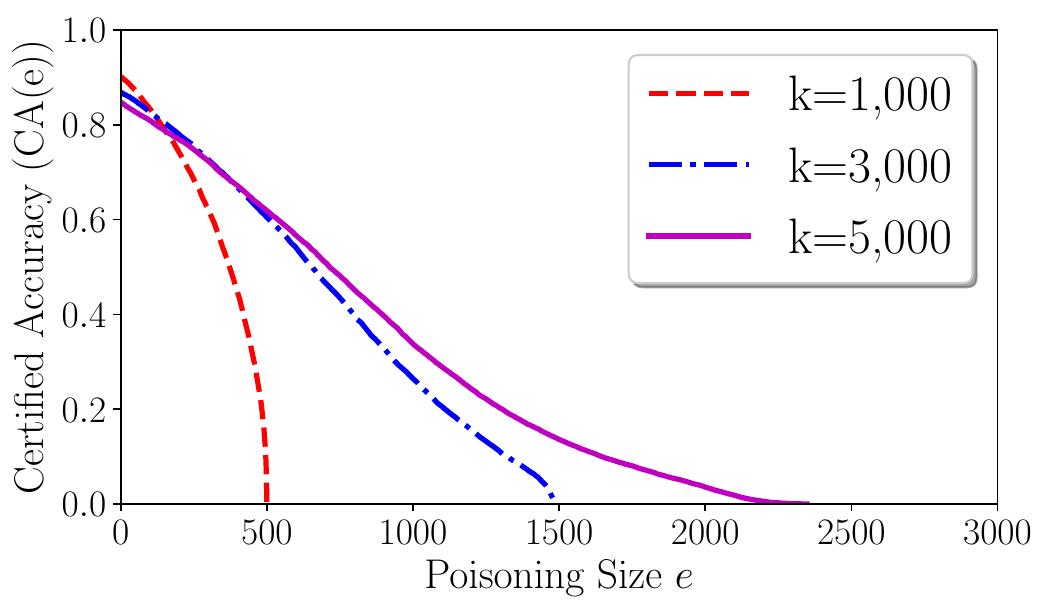}\label{impact_of_k_mnist}}
\subfloat[CIFAR10]{\includegraphics[width=0.45\textwidth]{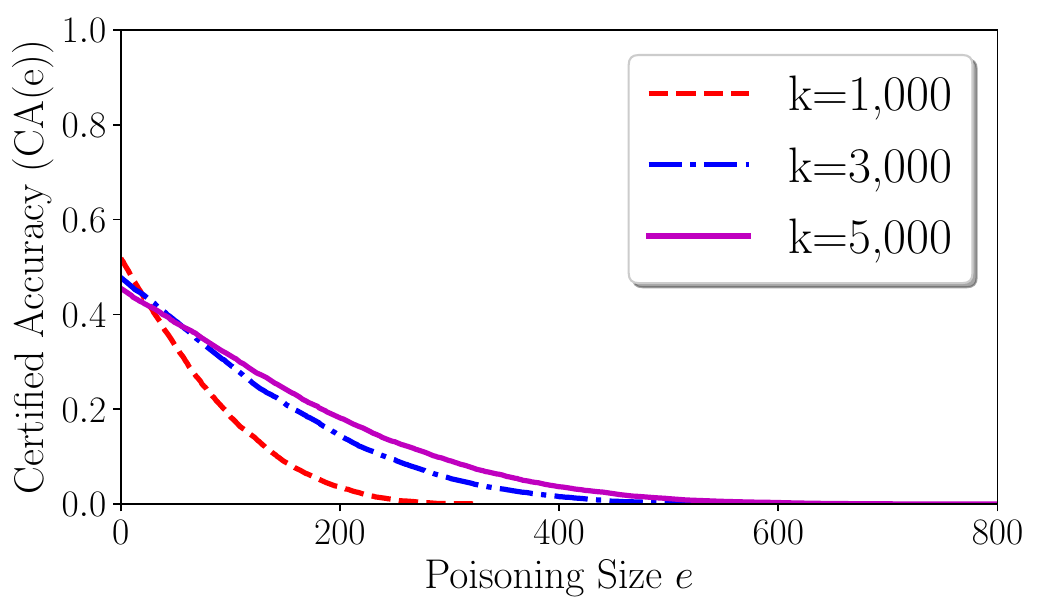}\label{impact_of_k_cifar10}}
	 \caption{Impact of $k$ on the certified accuracy of kNN against data poisoning attacks.} 
	 \label{impact_of_k_knn}
	 \vspace{-4mm}
\end{figure}

\begin{figure}[!t]
	 \centering
\subfloat[MNIST]{\includegraphics[width=0.45\textwidth]{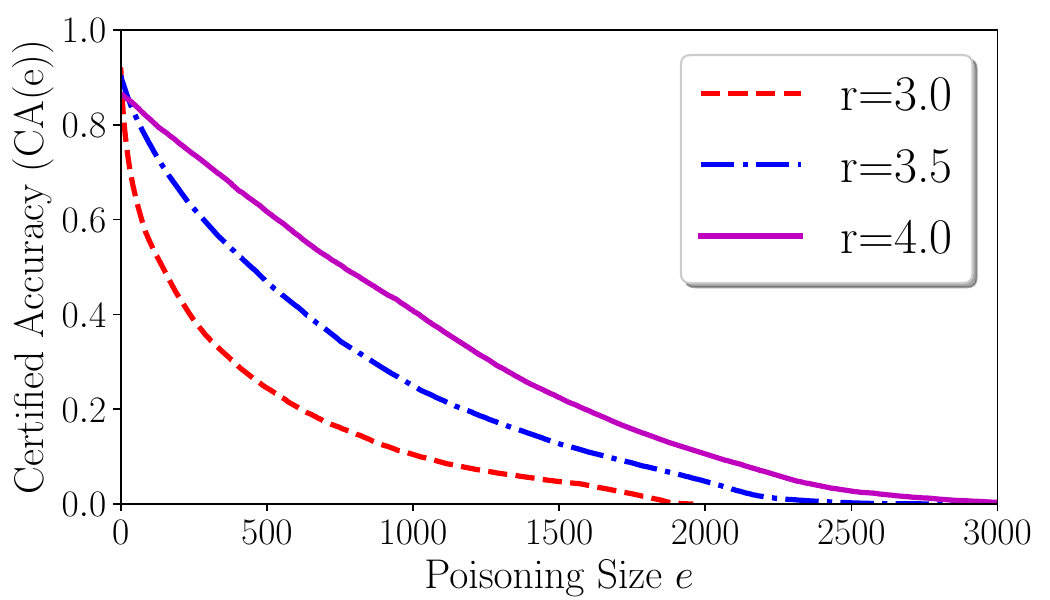}\label{impact_of_r_mnist}}
\subfloat[CIFAR10]{\includegraphics[width=0.45\textwidth]{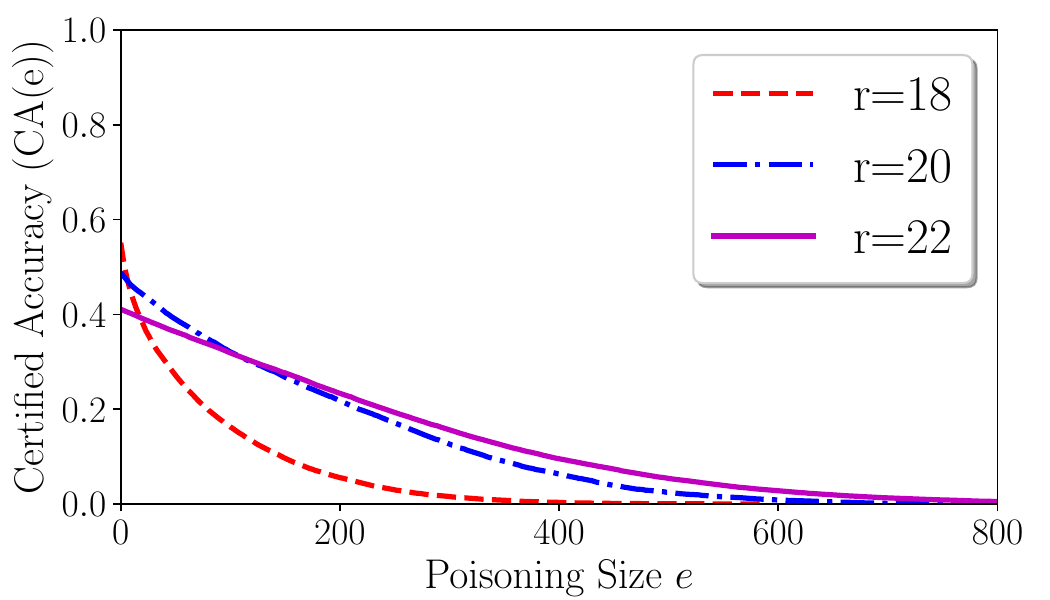}\label{impact_of_r_cifar10}}
	 \caption{Impact of $r$ on the certified accuracy of rNN against data poisoning attacks.} 
	 \label{impact_of_r_rnn}
	 \vspace{-4mm}
\end{figure}

\begin{figure}[!t]
	 \centering
\subfloat[Data poisoning attacks]{\includegraphics[width=0.45\textwidth]{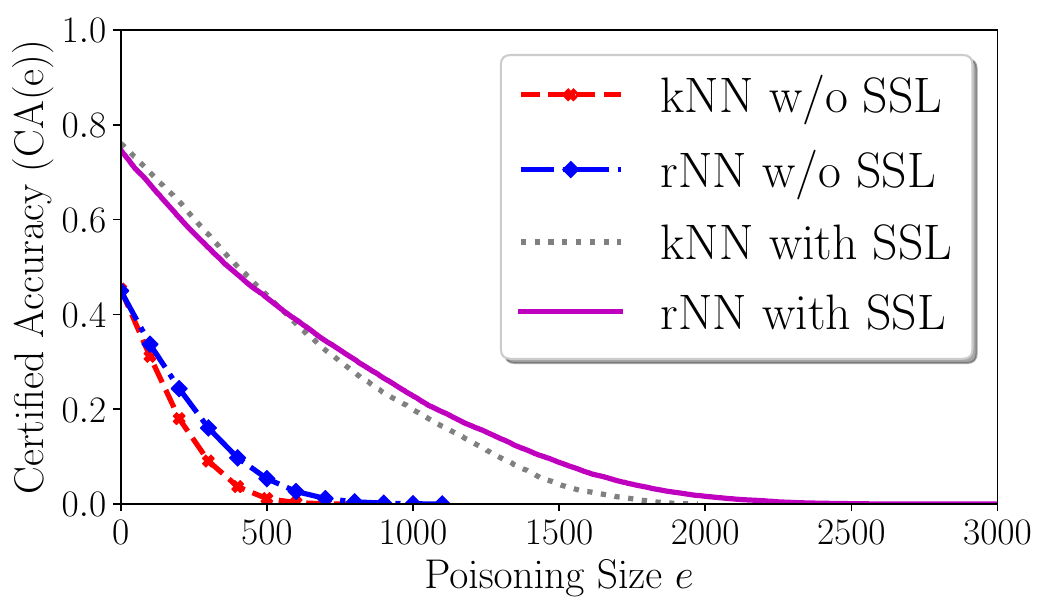}}
\subfloat[Backdoor attacks]{\includegraphics[width=0.45\textwidth]{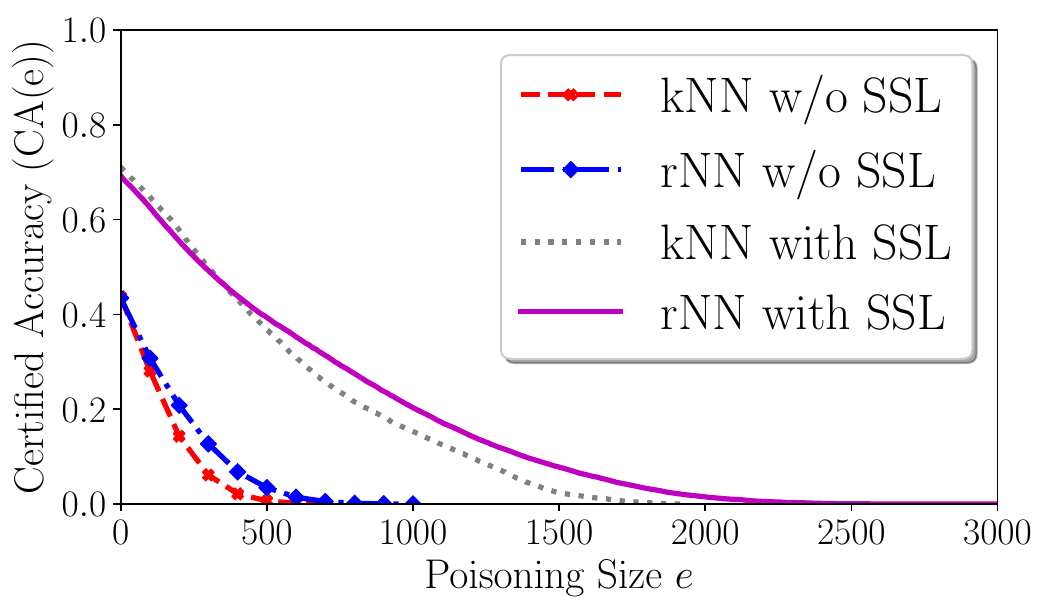}\label{impact_of_k_cifar10}}
	 \caption{Self-supervised learning improves the certified accuracy of kNN and rNN, where CIFAR10 is used.} 
	 \label{impact_of_tl}
	 \vspace{-4mm}
\end{figure}

\noindent
{\bf Self-supervised learning improves the certified accuracy:} 
\emph{Self-supervised learning (SSL)}~\citep{hadsell2006dimensionality,he2020momentum,chen2020simple} aims to learn a feature extractor using a large amount of unlabeled data, such that the feature extractor can be used to extract features for a variety of downstream learning tasks.  We adopt the pre-trained feature extractor called CLIP~\citep{radford2learning} to extract features. CLIP was pre-trained using 400 million (image, text) pairs and we use its public implementation~\citep{clip-code}. Note that we assume the pre-trained feature extractor is not poisoned. 
For each training or testing input in CIFAR10, we use CLIP to extract its features. Then, given the features, we use kNN or rNN to classify testing inputs.
We use the default $k$ (i.e., $k=5,000$) for both kNN without SSL and kNN with SSL. We adjust $r$ such that rNN without SSL (or rNN with SSL) has similar certified accuracy under no attacks with kNN without SSL (or kNN with SSL). Figure~\ref{impact_of_tl} shows the comparison results. Our results show that  self-supervised learning can significantly improve the certified accuracy of kNN and rNN.

%% file: related.tex
\section{Related Work}
\noindent
{\bf Data poisoning attacks and backdoor attacks:} Data poisoning attacks have been proposed against various learning algorithms such as   Bayes classifier~\citep{nelson2008exploiting}, SVM~\citep{biggio2012poisoning},   
neural networks~\citep{munoz2017towards,shafahi2018poison,suciu2018does,demontis2019adversarial,zhu2019transferable,huang2020metapoison}, 
recommender systems~\citep{li2016data,yang2017fake,fang2018poisoning,fang2020influence,huang2021data}, federated learning~\citep{bhagoji2019analyzing,fang2020local,bagdasaryan2020backdoor}, and others~\citep{rubinstein2009antidote,vuurens2011much}. Backdoor attacks~\citep{gu2017badnets,chen2017targeted,liu2017trojaning,jia2021badencoder}  corrupt both the training and testing phases of a machine learning system such that the corrupted classifier predicts an attacker-chosen label for any input embedded with a trigger.

\noindent
{\bf Defenses against data poisoning attacks and backdoor attacks:} To mitigate data poisoning attacks and/or backdoor attacks, many empirical defenses~\citep{cretu2008casting,rubinstein2009antidote,barreno2010security,biggio2011bagging,feng2014robust,jagielski2018manipulating,tran2018spectral,diakonikolas2019sever,liu2018fine,wang2019neural} have been proposed.~\citet{steinhardt2017certified} derived an upper bound of the loss function for data poisoning attacks when the model is learnt using examples in a feasible set.~\citet{diakonikolas2019sever} proposed a robust meta-algorithm which iteratively removes outliers such that the model parameters learnt on the poisoned dataset are close to those learnt on the clean dataset under certain assumptions.   
However, these defenses cannot guarantee that the predicted label for a testing input is certifiably unaffected and cannot provide (rigorous) certified accuracies.

Recently, several certified defenses~\citep{ma2019data,wang2020certifying,rosenfeld2020certified,weber2020rab,levine2020deep,jia2020intrinsic} were proposed to defend against data poisoning attacks and/or backdoor attacks. These defenses provide certified accuracies for a testing dataset either probabilistically~\citep{ma2019data,wang2020certifying,weber2020rab,jia2020intrinsic} or deterministically~\citep{rosenfeld2020certified,levine2020deep}. All these defenses except~\citet{ma2019data} leverage majority vote to predict the label of a testing example. In particular, a voter is a base classifier learnt on a perturbed version of the training dataset in randomized smoothing based defenses~\citep{rosenfeld2020certified,wang2020certifying,weber2020rab}, while a voter is a base classifier learnt on a subset of the trainig dataset in bagging~\citep{jia2020intrinsic} and DPA~\citep{levine2020deep}.~\citet{ma2019data} showed that a differentially private learning algorithm achieves certified accuracy against data poisoning attacks. They also train multiple differentially private classifiers, but they are not used to predict the label of a testing example via majority vote. Instead, their average accuracy is used to estimate the certified accuracy. 

kNN and rNN have intrinsic majority vote mechanisms and we show that they provide deterministic certified accuracies against data poisoning attacks and backdoor attacks. Moreover, rNN enables joint certification. We note that DPA~\citep{levine2020deep} proposed to use a hash function to assign training examples into partitions, which is different from our use of hash function. In particular, we use a hash function to rank training examples. Moreover, both DPA and our work rank the labels to break ties. 

\noindent
{\bf Nearest neighbors and robustness:} A line of works~\citep{wilson1972asymptotic,guyon1996discovering,peri2019deep,bahri2020deep} leveraged nearest neighbors to clean a training dataset. For instance,~\citet{wilson1972asymptotic} proposed to remove a training example whose label is not the same as the majority vote among the labels of its 3 nearest neighbors.~\citet{peri2019deep} proposed to remove a training example whose label is not the mode amongst labels of its k nearest neighbors in the feature space.~\citet{bahri2020deep} combined kNN with an intermediate layer of a preliminary deep neural network model to filter suspiciously-labeled training examples. 
Another line of works~\citep{gao2018consistency,reeve2019fast} studied the resistance of nearest neighbors to random noisy labels. For instance,~\citet{gao2018consistency} analyzed the resistance of kNN to asymmetric label noise and introduced a Robust kNN to deal with noisy labels.~\citet{reeve2019fast} further analyzed the Robust kNN proposed by~\citet{gao2018consistency} in the setting with unknown asymmetric label noise.   

 kNN and its variants have also been used to defend against adversarial examples~\citep{wang2018analyzing,sitawarin2019defending,papernot2018deep,sitawarin2019robustness,dubey2019defense,yang2020robustness,cohen2020detecting}. For instance,~\citet{wang2018analyzing} analyzed the robustness of nearest neighbors to adversarial examples and proposed a more robust 1-nearest neighbor.  Several works~\citep{amsaleg2017vulnerability,wang2018analyzing,wang2019evaluating,yang2020robustness} proposed adversarial examples to nearest neighbors, e.g.,~\citet{wang2019evaluating} proposed adversarial examples against 1-nearest neighbor. These works are orthogonal to ours as we focus on analyzing the certified robustness of kNN and rNN against data poisoning and backdoor attacks. 

%% file: conclusion.tex
\section{Conclusion and Future Work}
In this work, we derive the certified robustness of nearest neighbor algorithms, including kNN and rNN, against data poisoning attacks and backdoor attacks. Moreover, we derive a better lower bound of certified accuracy for rNN via jointly certifying multiple testing examples. 
Our evaluation results  show that 1) both kNN and rNN outperform state-of-the-art certified defenses against data poisoning attacks and backdoor attacks, and 2) joint certification outperforms individual certification.  Interesting future work includes 1) extending joint certification to other learning algorithms, 2) improving joint certification via new grouping methods, and 3) improving certified accuracy of kNN and rNN via new distance metrics.

%% file: appendix.tex
\appendix

\clearpage

\section{Proof of Theorem~\ref{nn_certified_theorem}}
\label{proof_of_theorem_of_1}
\begin{proof}
When an attacker poisons $S(D_{tr},D_{tr}^*)$ training examples, the number of changed nearest neighbors in $\mathcal{N}(D_{tr},\mathbf{x})$ is at most $S(D_{tr},D_{tr}^*)$. We denote by  $s_{l}^{*} = \sum_{(\mathbf{x}_j, y_j)\in \mathcal{N}(D^{*}_{tr},\mathbf{x})} \mathbb{I}(l = y_j)$ the number of votes for label $l$ among the nearest neighbors $\mathcal{N}(D^{*}_{tr},\mathbf{x})$ in the poisoned training dataset, where $l=1,2,\cdots,c$. Then, we have  $s_l - S(D_{tr},D_{tr}^*) \leq s_{l}^{*}\leq s_l + S(D_{tr},D_{tr}^*)$ for each $l=1,2,\cdots,c$. Therefore, when $S(D_{tr},D_{tr}^*) \leq \lceil \frac{s_a - s_b + \mathbb{I}(a>b)}{2} \rceil - 1$, we have $s_{a}^{*} - s_{b}^{*} \geq s_a - s_b - 2 \cdot S(D_{tr},D_{tr}^*) > 0$ if $a < b$ and $s_{a}^{*} - s_{b}^{*} \geq s_a - s_b - 2 \cdot S(D_{tr},D_{tr}^*) \geq 0$ if $a>b$. Thus, the nearest neighbor algorithm still predicts label $a$ for $\mathbf{x}$ in both cases based on our way of breaking ties, i.e., we have $\mathcal{M}(D_{tr}^*, \mathbf{x})=a$ when $S(D_{tr},D_{tr}^*) \leq \lceil \frac{s_a - s_b + \mathbb{I}(a>b)}{2} \rceil - 1$. 
\end{proof}

\section{Proof of Theorem~\ref{theorem_of_ic_aggregate}}
\label{proof_of_theorem_of_ic_aggregate}
\begin{proof}
We have the following: 
{\small 
\begin{align}
&CA(e) \\
=& \min_{D^{*}_{tr}, S(D_{tr}, D^{*}_{tr}) \leq e} \frac{\sum_{(\mathbf{x}_i,y_i)\in D_{te}}\mathbb{I}(\mathcal{M}(D^{*}_{tr},\mathbf{x}_i)= y_{i})}{|D_{te}|} \\ 
\geq&\frac{\sum_{(\mathbf{x}_i,y_i)\in D_{te}} \min\limits_{D^{*}_{tr}, S(D_{tr}, D^{*}_{tr}) \leq e} \mathbb{I}(\mathcal{M}(D^{*}_{tr},\mathbf{x}_i)= y_{i})}{|D_{te}|} \\
=& \frac{\sum\limits_{(\mathbf{x}_i,y_i)\in D_{te}} \mathbb{I}(a_i=y_i) \min\limits_{D^{*}_{tr}, S(D_{tr}, D^{*}_{tr}) \leq e}  \mathbb{I}(\mathcal{M}(D^{*}_{tr},\mathbf{x}_i)= a_i)}{|D_{te}|} \\
=&\frac{\sum_{(\mathbf{x}_i,y_i)\in D_{te}} \mathbb{I}(a_i=y_i) \mathbb{I}(e \leq e_i^{*})}{|D_{te}|},
\end{align}
}
where the last step is based on applying Theorem~\ref{nn_certified_theorem} to testing input $\mathbf{x}_i$.

\end{proof}

\section{Proof of Theorem~\ref{joint_dataset_aggregate_opt}}
\label{proof_of_aggregate_opt}
\begin{proof}

We have the following: 
\begin{align}
&CA(e) \\
=& \min_{D^{*}_{tr}, S(D_{tr}, D^{*}_{tr})\leq e} \frac{\sum_{(\mathbf{x}_i,y_i)\in D_{te}}\mathbb{I}(\mathcal{M}(D^{*}_{tr},\mathbf{x}_i)= y_{i})}{|D_{te}|} \\ 
=& \min_{D^{*}_{tr}, S(D_{tr}, D^{*}_{tr})\leq e} \frac{\sum_{j=1}^{\lambda} \sum_{(\mathbf{x}_i,y_i)\in \mathcal{U}_j}\mathbb{I}(\mathcal{M}(D^{*}_{tr},\mathbf{x}_i)= y_{i})}{\sum_{j=1}^{\lambda}|\mathcal{U}_j|} \\ 
\label{entire_dataset_proof_step_1}
\geq&\frac{\sum_{j=1}^{\lambda} \min\limits_{D^{*}_{tr}, S(D_{tr}, D^{*}_{tr})\leq e}\sum_{(\mathbf{x}_i,y_i)\in \mathcal{U}_j} \mathbb{I}(\mathcal{M}(D^{*}_{tr},\mathbf{x}_i)= y_{i})}{\sum_{j=1}^{\lambda}|\mathcal{U}_j|} \\
\label{entire_dataset_proof_step_2}
\geq&\frac{\sum_{j=1}^{\lambda} \mu_j \cdot |\mathcal{U}_j|}{\sum_{j=1}^{\lambda}|\mathcal{U}_j|},
\end{align}
where we have Equation~(\ref{entire_dataset_proof_step_2}) from~(\ref{entire_dataset_proof_step_1}) based on applying Theorem~\ref{nn_certified_theorem_opt} to group $\mathcal{U}_j$. 
\end{proof}

\begin{figure}[!t]
	\vspace{-2mm}
	 \centering
\subfloat[MNIST]{\includegraphics[width=0.45\textwidth]{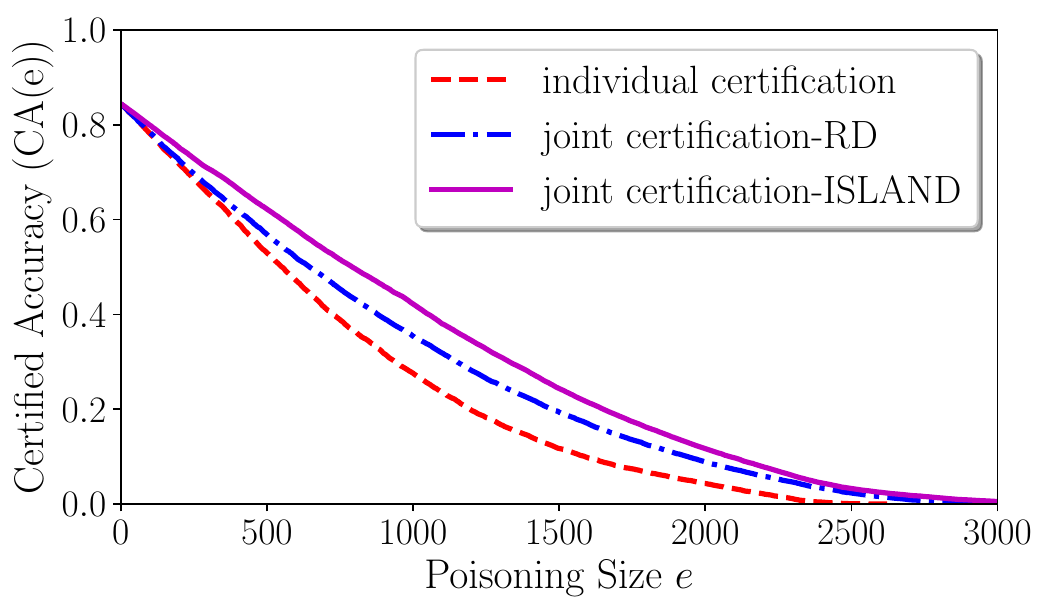}\label{impact_of_group_mnist_ba}}
\subfloat[CIFAR10]{\includegraphics[width=0.45\textwidth]{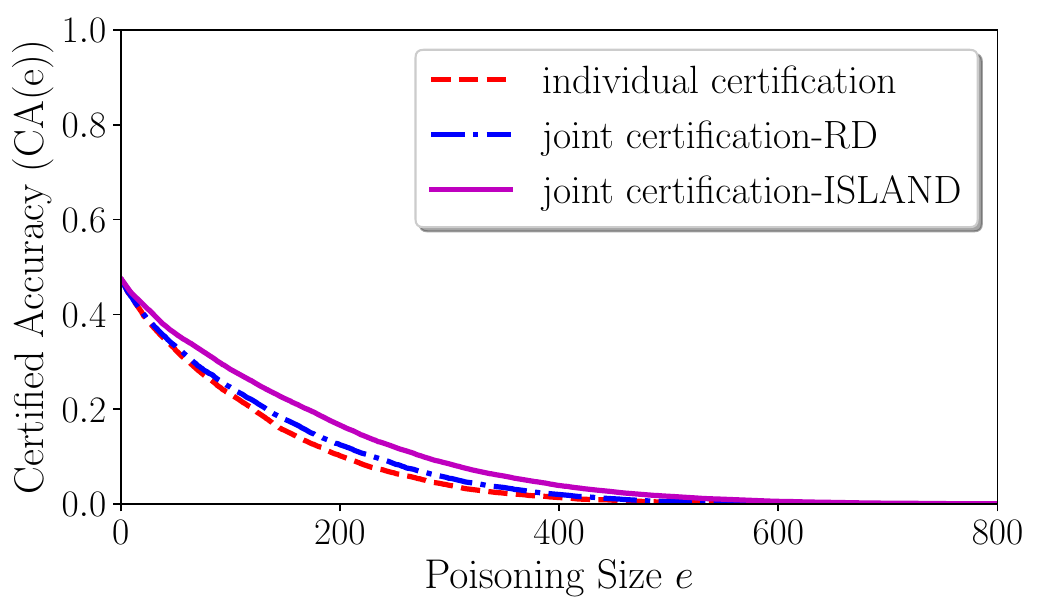}\label{impact_of_group_cifar10_ba}}
	 \caption{Comparing individual certification with joint certification for rNN against backdoor attacks. } 
	 \label{compare_hog_feature_extractor_joint_ba}
\end{figure}

\begin{figure}[!t]
	\vspace{-2mm}
	 \centering
\subfloat[MNIST]{\includegraphics[width=0.45\textwidth]{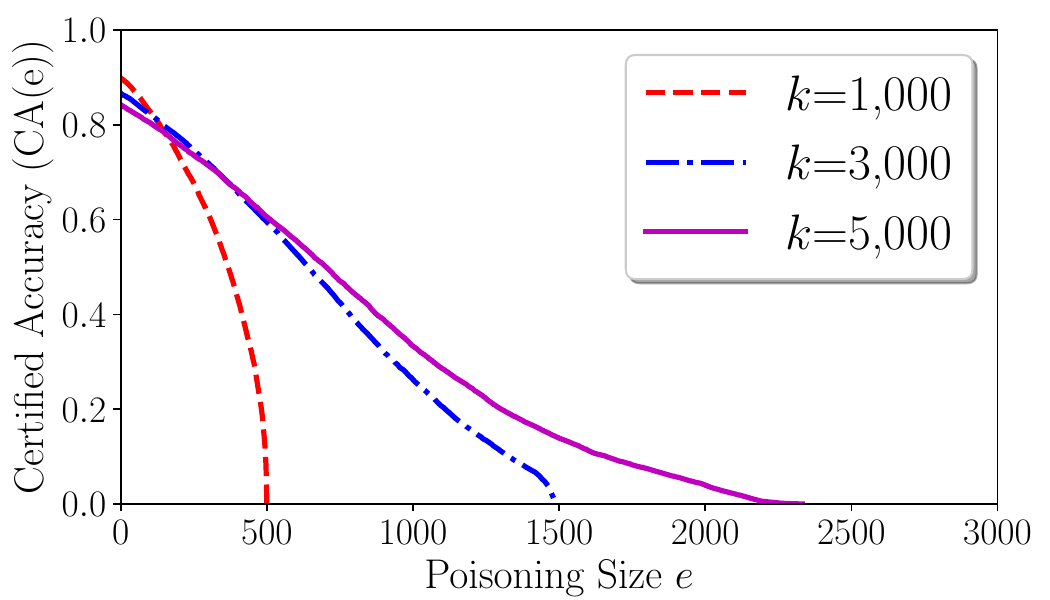}\label{impact_of_k_mnist_ba}}
\subfloat[CIFAR10]{\includegraphics[width=0.45\textwidth]{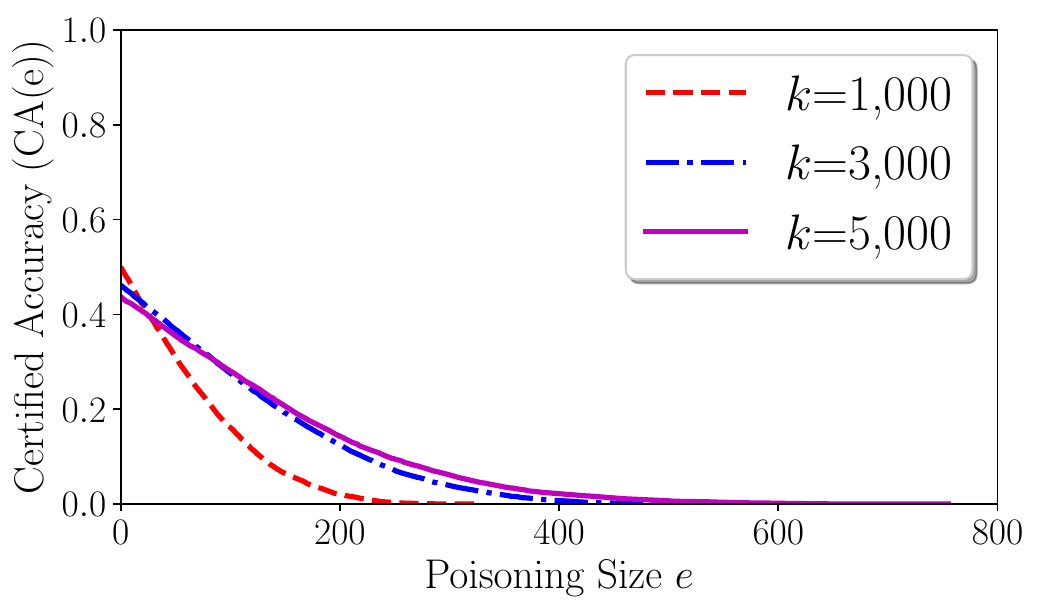}\label{impact_of_k_cifar10_ba}}
	 \caption{Impact of $k$ on the certified accuracy of kNN against backdoor attacks. } 
	 \label{impact_of_k_knn_ba}
\end{figure}

\begin{figure}[!t]
	\vspace{-2mm}
	 \centering
\subfloat[MNIST]{\includegraphics[width=0.45\textwidth]{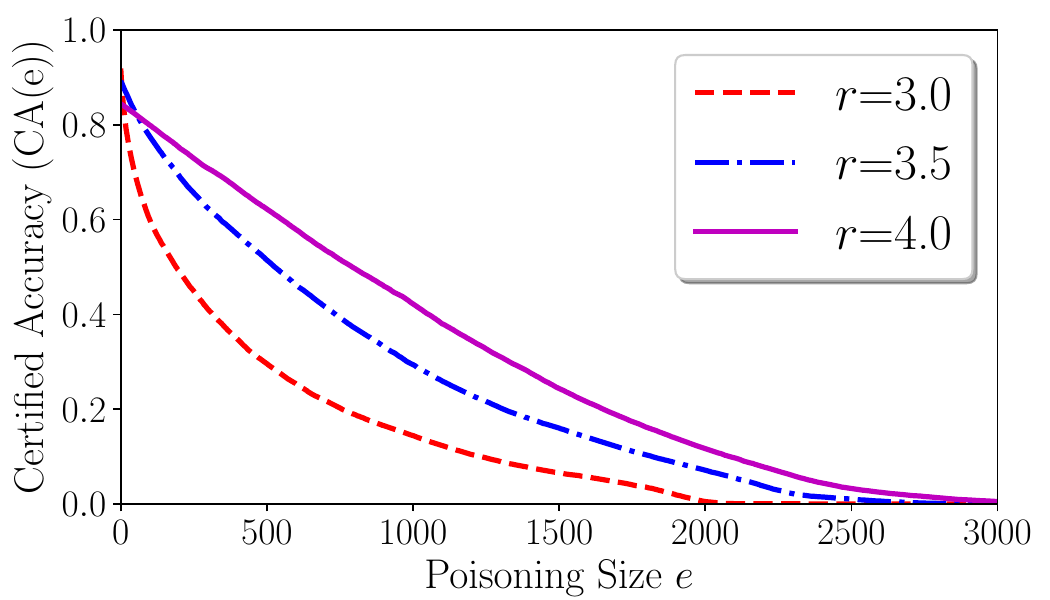}\label{impact_of_r_mnist_ba}}
\subfloat[CIFAR10]{\includegraphics[width=0.45\textwidth]{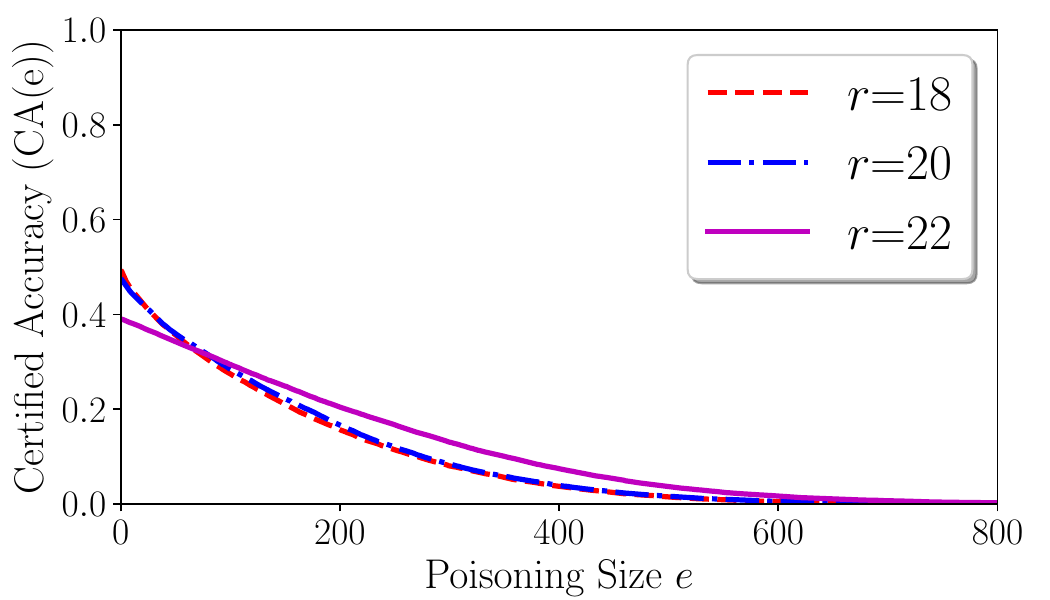}\label{impact_of_r_cifar10_ba}}
	 \caption{Impact of $r$ on the certified accuracy of rNN against backdoor attacks. } 
	 \label{impact_of_r_rnn_ba}
\end{figure}